\documentclass[a4paper]{article}
\usepackage{amsmath,amssymb,amsthm,fullpage}
\usepackage[shortcuts]{extdash}
\usepackage{hyperref}

\usepackage{csquotes}
\usepackage[backend=biber,maxbibnames=99,sorting=none]{biblatex}
\bibliography{cymoduli.bib}

\usepackage{color}
\usepackage{lipsum}
\usepackage{amssymb,amsmath,amsfonts,epic,eepic,float}

\usepackage[mathscr]{eucal}

\usepackage{rotating,epsfig,indentfirst,array,varioref}
\usepackage{appendix,marginnote,tikz,pgf,mathtools}
\usepackage{setspace}

\usepackage{authblk}

\makeatletter
\def\@@bfil{\leaders \vrule \@height \ht\z@ \@depth \z@ \hfill}
\def\@bLfil{\@@bfil}
\def\@bRfil{\@@bfil}
\def\resetbraceratio{\gdef\@bLfil{\@@bfil}\gdef\@bRfil{\@@bfil}}
\def\setbraceratio#1#2{
  \let\@bLfil\relax
  \multido{\iA=1+1}{#1}{\gappto\@bLfil{\@@bfil}}
  \let\@bRfil\relax
  \multido{\iA=1+1}{#2}{\gappto\@bRfil{\@@bfil}}
}
\def\upbracefill{$\m@th\setbox\z@\hbox{$\braceld$}\bracelu\@bLfil\bracerd\braceld\@bRfil\braceru$}
\def\downbracefill{$\m@th\setbox\z@\hbox{$\braceld$}\braceld\@bLfil\braceru\bracelu\@bRfil\bracerd$}
\makeatother

\usepackage{tikz,pgf}
\usetikzlibrary{shapes}
\usetikzlibrary{calc}
\usetikzlibrary{decorations.pathmorphing}
\usetikzlibrary{decorations.pathreplacing,shapes.misc}
\usetikzlibrary{positioning}
\usetikzlibrary{arrows}
\usetikzlibrary{decorations.markings}

\def\be{\begin{equation}}
\def\ee{\end{equation}}
\def\barr{\begin{array}}
\def\earr{\end{array}}

\newtheorem{prop}{Proposition}

\def\1{\tilde{1}}
\def\2{\tilde{2}}
\def\3{\tilde{3}}

\newtheorem{lemma}{Lemma}[section]

\newcommand{\ba}{\begin{equation}\begin{aligned}}
\newcommand{\ea}{\end{aligned}\end{equation}}

\newcommand{\bml}{\begin{multline}}
\newcommand{\eml}{\end{multline}}

\newcommand{\CC}{\mathbb{C}}
\newcommand{\RR}{\mathbb{R}}

\newcommand{\ZZ}{\mathbb{Z}}
\newcommand{\PP}{\mathbb{P}}

\newcommand{\dd}{\mathrm{d}}
\newcommand{\pd}{\partial}

\newcommand{\Res}{\mathrm{Res}}

\newcommand{\MM}{\mathscr{M}}
\newcommand{\aA}{\mathfrak{A}}

\mathtoolsset{showonlyrefs}

\begin{document}

\title{Special geometry on the 101 dimesional moduli space of the quintic threefold.}

\author{Konstantin Aleshkin $^{1,2}$,}
\author{Alexander Belavin $^{1,3}$ }

\affil{$^1$ L.D. Landau Institute for Theoretical Physics\\
 Akademika Semenova av. 1-A\\ Chernogolovka, 142432  Moscow region, Russia}

\affil{$^2$ International School of Advanced Studies (SISSA),
 via Bonomea 265, 34136 Trieste, Italy}
\affil{$^3$ Moscow Institute of Physics and Technology\\
Dolgoprudnyi, 141700 Moscow region, Russia}

\maketitle
\abstract{ A new method for explicit computation of the CY moduli space metric was
proposed by the authors recently. The method makes use of the connection of the moduli space with
a certain Frobenius algebra. Here we clarify this approach and demonstrate its efficiency by computing 
the Special geometry  of the 101-dimensional moduli space of the quintic threefold around the orbifold point.}

\flushbottom




\section{Introduction.}
 When compactifying the IIB superstring theory on a Calabi--Yau (CY) threefold $X,$ one can
write the low-energy effective theory in terms of the geometry of the CY moduli space~\cite{CHSW}.
More precisely, the effective Lagrangian of the vector multiplets in the superspace contains
 $h^{2,1}$ supermultiplets. Scalars from these multiplets take value in the
 target space $\MM$, which is a moduli space
 of complex structures on a CY manifold and is a special K\"ahler manifold itself~\cite{Rolling, 
CO, S}.
 Metric $G_{a\bar{b}}$ and Yukawa couplings $\kappa_{abc}$ 
on this space are given by the following formulae:
\ba \label{koo}
&G_{a\bar{b}} = \pd_a\pd_{\bar{b}} K, \;\;\;
 e^{-K} = -i \int_X \Omega\wedge\bar{\Omega},\\
&\kappa_{abc} = \int_X \Omega\wedge\pd_a\pd_b\pd_c\Omega = \frac{\pd^3 F}{\pd z^a \pd z^b \pd z^c},
\ea
where
\be
z^a = \int_{A_a} \Omega, \; \frac{\pd F}{\pd z^a} = \int_{B^a} \Omega
\ee
are the period integrals of the holomorphic volume form $\Omega$ on $X$.
Here $A_a$ and $B^a$ form the symplectic basis in $H_3(X,\ZZ)$.
We can rewrite the expression~\eqref{koo} for the K\"ahler potential using the periods as
\be \label{ksympl}
e^{-K} = -i \Pi \Sigma \Pi^{\dagger}, \; \Pi = (\pd F, \; z), 
\ee
where matrix $(\Sigma)^{-1}$ is an intersection matrix of cycles $A_a, \; B^a$ equal to the
symplectic unit.
In practice, computation of periods in the symplectic basis is a very complicated problem
and was done explicitly only in few examples~\cite{COGP, Klemm, CFKM, COFKM}. It is due to the fact,
that it requires a case by case analysis and geometric description of the symplectic basis of cycles.
Recently we proposed a method~\cite{AKBA, AKBA2} to easily compute the K\"ahler metric (and the symplectic basis)
for a large class of CY manifolds which can be represented by specific hypersurfaces in
weighted projective spaces~\cite{BerHub}. Our method does not require the knowledge of symplectic cycles, but
instead uses a structure of a Frobenius algebra  associated with a CY of this 
class and its Hodge structure.~\footnote{Actually, moduli space of a CY manifold is closely related with
a Frobenius manifold~\cite{Dub}, and the Frobenius algebra, we use, is a tangent space to this manifold
at one point. }

Namely, let a CY manifold $X$ be given as a solution of an equation 
\be
W(x, \phi) = W_0(x) + \sum_{s=1}^{h^{2,1}} \phi_s e_s(x)  = 0
\ee
in  some weighted projective space, where $W_0(x)$ is a quasihomogeneous function in $\CC^5$
 of weight $d$ that defines an isolated singularity at $x=0,$(see~\cite{AVG}) which is tightly related with
the underlying $N=2$ superconformal theory~\cite{LVW, Mart, Gep}. 
The monomials $e_s(x) $ also have  weight $d$ and correspond to deformations of
 the complex structure of  $X$.

Polynomial $W_0(x)$ defines a Milnor ring $R_0$. Inside  $R_0$  there exists a subring $R^Q_0$  which
 is invariant w.r.t the action of the so-called quantum symmetry group $Q$.
This group acts on $\CC^5$ diagonally, and preserves $W(x, \phi)$. In many cases 
$\dim R_0^Q = \dim H^3(X)$ and the ring itself has a Hodge structure $R_0^Q = (R_0^Q)^0 
\oplus (R_0^Q)^1 \oplus (R_0^Q)^2 \oplus (R_0^Q)^3$ in correspondence with degrees of the elements 
$0,d,2d,3d$. One can introduce an invariant pairing  $\eta$ on $R_0^Q$. The pairing turns the ring to a 
Frobenius algebra \cite{Dub} and plays an important role in the construction of our formula for $e^{-K}$.

Using the invariant ring $R_0^Q$ and differentials $D_{\pm} = \dd \pm \dd W_0\wedge$ we construct
two groups of $Q-$invariant cohomology $H^5_{D_{\pm}}(\CC^5)_{inv}$. 
These groups inherit the Hodge
structure from $R_0^Q$. If we denote by $\{e_{\mu}(x) \}$ some basis of $R_0^Q,$ then
$\{e_{\mu}(x) \, \dd^5 x\}$ will be a basis of $H^5_{D_{\pm}}(\CC^5)_{inv}$.
As shown by Candelas~\cite{Candelas}, elements of these cohomology groups are in correspondence
 with harmonic
forms of $H^3(X)$. This isomorphism sends components $H^{3-q,q}(X)$ to the Hodge
decomposition components of $H^5_{\pm}(\CC^5)_{inv}$ spanned by $e_{\mu}(x) \, \dd^5 x$ with
$e_{\mu}(x) \in (R_0^Q)^{q}$ and sends the pairing on the differential forms on $X$
to the invariant Frobenius algebra pairing $\eta$. Also the same isomorphism allows to define a complex conjugation 
(we denote this operation $*$) on the invariant cohomology. 

It turns out, that in the basis $\{e_{\mu}(x)\}$ the operation $*$ reads
\be \label{conj}
*e_{\mu}(x)\,\dd^5 x =  M_{\mu}^{\nu} e_{\nu} (x) \, \dd^5 x 
\ee
If we pick a basis $\{ e_{\mu}(x)\}$ such, that the Frobenius pairing $\eta$ is antidiagonal,
the matrix $M$ is antidiagonal as well:
\be
 *e_{\mu}(x)\,\dd^5 x = A^{\mu} \, e_{\mu'}(x)\, \dd^5 x,
\ee
where $e_{\mu'}(x)$ is the unique element of the basis such, that 
$\eta(e_{\mu}(x), e_{\mu'}(x)) = 1,$ that is $e_{\mu}(x) \cdot e_{\mu'}(x) = 
e_{\rho}(x),$ which is a unique up to a constant element of degree 3d. 
Coefficients $A_{\mu}$ are the coefficients of the matrix $M$ in this basis. In particular, we have
a useful relation: $A_{\mu} \overline{A_{\mu'}}=1,$ because $M$ is an 
anti-involution.
We compute $A_{\mu}$ in the section~\ref{sec:comp} for the quintic threefold, see also section~\ref{sec:real}.

Having $H^5_{D_{\pm}}(\CC^5)_{inv}$ we define the relative invariant homology groups
$\mathscr{H}_{5}^{\pm,inv} := H_5(\CC^5, W_0 = L,\;\mathrm{Re}L \to \pm \infty)_{inv}$ inside a relative homology group
$H_5(\CC^5, W_0 = L,\;\mathrm{Re}L \to \pm \infty)$. For this purpose we use oscillatory integrals. Using 
the oscillatory integral pairing we define a cycle 
$\Gamma^{\pm}_{\mu}$  in the basis of relative invariant homology 
to be  dual to  $e_{\mu}(x) \, \dd^5 x$.

At last we define periods $\sigma^{\pm}_{\mu}(\phi)$ to be oscillatory integrals over the basis 
of cycles $\Gamma^{\pm}_{\mu}$, which can be effectively computed using the techinque described 
in~\cite{BB, AKBA, AKBA2} that we remind in the section~\ref{sec:osc}.
 The periods $\sigma^{\pm}_{\mu}(\phi)$ are equal to periods of the holomorphic volume
 form $\Omega$
on $X$ in a special basis of cycles $H_3(X, \CC)$ with complex coefficients. 

As shown in~\cite{AKBA}, K\"ahler potential for the metric is given by the followng formula 
\be \label{intkah}
e^{-K(\phi)} = \sum_{\mu,\nu, \lambda}  \sigma^+_{\mu}(\phi) \eta^{\mu\lambda}
 M^{\nu}_{\lambda} \overline{\sigma^-_{\nu}(\phi)},
\ee
where the real structure matrix $M$ is the same as the one in~\eqref{conj}.
Matrix $M$ can be also represented as $M = T^{-1} \bar{T},$ where $T$ is a transition
 matrix from periods
in arbitrary real basis of cycles $Q^{\pm}_{\mu}$ to periods $\sigma^{\pm}_{\mu}(\phi)$.
In our basis matrix $\eta$ is antidiagonal, and it follows, that
\be \label{intdiag}
e^{-K(\phi)} = \sum_{\mu}  (-1)^{|\nu|}\sigma^+_{\mu}(\phi) A^{\mu}
 \overline{\sigma^-_{\mu}(\phi)}.
\ee
Using this we are able to explicitly compute the diagonal matrix elements  $A^{\mu}$ and to obtain the explicit expression  for the whole $e^{-K}$.

In~\cite{AKBA, AKBA2}, to find the real structure,
 we used the knowledge of periods in some integral basis of homology cycles 
(e.g. from~\cite{BCOFHJQ}). However this basis is not
always known.\\
 In this paper we propose another method to compute the real structure matrix $M$ and apply it to
 the $101$-dimensional moduli space of the quintic threefold complex structures 
around the orbifold point to get an explicit exact result for the moduli space K\"ahler metric. Together with the knowledge of  the geometry of the $1$-dimensional
 moduli space of   the quintic  K\"ahler structures computed via the mirror symmetry in~\cite{COGP} 
it presumably gives the geometry of
the full moduli space of Calabi-Yau quintic threefold. 

In what follows we apply our method for the quintic threefold, where the huge symmetry group
$S_5 \ltimes (\ZZ_5)^5$ simplifies the computations.
$(\ZZ_5)^5$ is called a group of phase symmetries, it acts diagonally on $\CC^5$ and preserves
 $W_0(x)$. It acts naturally on the invariant ring $R_0^Q$, and this action respects the Hodge
 decomposition of  $R_0^Q$. 
This allows to pick a basis $e_{\mu}(x)$ in each of the Hodge decomposition components
 of $R_0^Q$, which consists of eigenvectors of the phase symmetry group action, which simplifies
our computations. The $S_5,$ which acts by permutations of $x_i$ among themselves, allows to
reduce the amount of computations even further.

If we consider other hypersurfaces in weighted projective spaces, they have
 less symmetry then the quintic threefold.
 However, most of the considerations are true and allow to perform explicit computations
in the more general case, as we briefly discuss in the 
conclusion~\ref{sec:concl}.

\section{Hodge structure on the middle cohomology of the quintic}

First of all we notice, that the formula~\eqref{ksympl} may be written in the arbitrary basis of
cycles $q_{\mu} \in H_3(X, \ZZ)$:
\ba
e^{-K(\phi)} &= \omega_{\mu}(\phi) C^{\mu\nu} \overline{\omega_{\nu}(\phi)}, \\
 \omega_{\mu}(\phi) &= \int_{q_{\mu}} \Omega
\ea
and $(C^{-1})_{\mu\nu} = q_{\mu}
\cap q_{\nu}$.\\

Now let us specialize to the case where $X$ is a quintic threefold:
\be
X = \{(x_1: \cdots : x_5) \in \PP^4 \; | \; W(x, \phi) = 0 \},
\ee
where
\be \label{W}
W(x, \phi) = W_0(x) + \sum_{t=0}^{100} \phi_t e_t(x), \; W_0(x) = x_1^5+x_2^5+x_3^5+x_4^5+x_5^5  
\ee
and $e_t(x)$ are the degree 5 monomials such that each variable has the power that is a
 non-negative integer  less then four.
Let us denote monomials $e_t(x) = x_1^{t_1} x_2^{t_2}x_3^{t_3} x_4^{t_4}x_5^{t_5}$ by its degree
vector $t = (t_1, \cdots,  t_5)$. Then there are precisely 101 of such monomials, which can be 
divided into $5$ sets in respect to the permutation group $S_5$: 
$(1,1,1,1,1),$ $(2,1,1,1,0),$ $(2,2,1,0,0),$ $(3,1,1,0,0),$ $(3,2,0,0,0)$. 
In these groups there are correspondingly 1, 20, 30, 30, 20 different
monomials.
We denote $e_0(x) := e_{(1,1,1,1,1)}(x) = x_1x_2x_3x_4x_5$ to be the so-called
 fundamental monomial, which will be somewhat distinguished in our picture. \\

For this CY $\dim H_3(X) = 204$ and period integrals have the form
\be
\omega_{\mu}(x) = \int_{q_{\mu}} \frac{x_5 \, \dd x_1\dd x_2\dd x_3}{\pd W(x, \phi)/\pd x_4} = 
\int_{Q_{\mu}} \frac{\dd x_1 \cdots \dd x_5}{W(x, \phi)},
\ee
where $q_{\mu}\in H_3(X,\ZZ)$  and $Q_{\mu} \in H_5(\CC^5 \backslash (W(x, \phi) = 0), \ZZ)$ 
are the corresponding cycles. 
Cohomology groups of a K\"ahler manifold possess a Hodge structure $H^3(X) = H^{3,0}(X)\oplus
H^{2,1}(X)
\oplus H^{1,2}(X)\oplus H^{0,3}(X)$.
Period integrals measure variation of the Hodge structure on $H^3(X)$ as the complex structure
 on $X$ varies with $\phi$.
 This Hodge structure variation is equivalent to the one on a certain ring which we
will now describe.\\

\section{Hodge structure on the invariant Milnor ring.}

We can consider $W_0(x)$ as a singularity in $\CC^5$. Then there is an associated Milnor (also Jacobi)
 ring
\be
R_0 = \frac{\CC[x_1, \cdots, x_5]}{\langle\pd_i W \rangle}.
\ee
We will identify its elements with unique smallest degree polynomial representatives.
For the quintic threefold $X$ its Milnor ring $R_0$ is generated as a vector space
 by monomials where each variable 
has degree less than four, and $\dim R_0 = 1024$. Polynomial $W_0(x)$ is homogeneous and, in particular,
$W_0(\alpha  x_1, \ldots, \alpha x_5) = W_0(x_1, \ldots, x_5)$ for $\alpha^5 = 1$. This action preserves $W_0(x)$ and is trivial in the corresponding projective space and on $X$. Such a group with this action is called 
a~\textit{quantum symmetry} $Q$, in our case $Q \simeq \ZZ_5$. 
$Q$ obviously acts on the Milnor ring $R_0$. 

Now we define a subring $R_0^Q$ in the  Milnor ring $R_0$,
\be
R_0^Q := \{e_{\mu}(x) \in R_0 \; | \; e_{\mu}(\alpha x)  = e_{\mu}(x)\}, \; \alpha^5 = 1,
\ee
to be a $Q$-invariant part of the Milnor ring. 

It is multiplicatively
generated by 101 fifth-degree monomials $e_t(x)$ from~\eqref{W}. More precisely,
$R_0^Q$ consists of elements of degree $0,5,10$ and $15$, dimensions of the corresponding subspaces
are $1,101,101$ and $1$. This degree filtration defines a Hodge structure on $R^Q_0$. 
Basically  $R_0^Q$ is isomorphic to $H^3(X)$ and the isomorphism  sends the degree filtration to the Hodge filtration on $H^3(X)$ \cite{Candelas}. 
Let us denote $\chi^i_{\bar{j}} = g^{i\bar{k}} \, \chi_{\bar{k}\bar{j}}$ as an extrinsic curvature tensor for the hypersurface $W(x, \phi) = 0$ in $\PP^4$. 
Then the isomorphism above can be written as a map from $R_0^Q$ to closed differential
forms in $H^3(X)$:

\ba \label{chi}
&1 \to \Omega_{ijk} \in H^{3,0}(X), \\
&e_{\mu}(x) \to e_{\mu}(x(y)) \, \chi^l_{\bar{i}} \, \Omega_{ljk} \in H^{2,1}(X) \text{ if } |\mu| = 5, \\
&e_{\mu}(x) \to e_{\mu}(x(y)) \, \chi^l_{\bar{i}} \, \chi^m_{\bar{j}} \, \Omega_{lmk} \in H^{1,2}(X) \text{ if }
 |\mu| = 10, \\
&e_{\rho}(x) = x_1^3x_2^3x_3^3x_4^3x_5^3 \to \chi^l_{\bar{i}} \, \chi^m_{\bar{j}} \chi^p_{\bar{k}} \, \Omega_{lmp} = \kappa \bar{\Omega} \in H^{0,3}(X)
\ea
 The details on this map can be found in~\cite{Candelas, CanKal}. We also introduce the notation $e_{\mu}(x)$  for elements of the monomial basis of $R^Q_0$,
where  $\mu = (\mu_1, \cdots, \mu_5), \; \mu_i \in \ZZ_+^5 ,\; e_{\mu}(x) = \prod_i x_i^{\mu_i}$ and $|\mu| = \sum \mu_i$ is  the degree of $e_{\mu}(x)$. 
In particular, $\rho = (3,3,3,3,3),$ that is $e_{\rho}(x)$ is a unique degree 15 element of $R_0^Q$. 

There is a $\ZZ^5_5$ phase symmetry group acting diagonally on $\CC^5$:
$\alpha \cdot (x_1, \cdots, x_5) = (\alpha_1 x_1, \cdots, \alpha_5 x_5), \; \alpha_i^5=1$.
 This action preserves
$W_0 = \sum_i x_i^5$. The mentioned above quantum symmetry $Q$ is a diagonal
 subgroup of the phase symmetries. Basis $\{e_{\mu}(x)\}$ is an eigenbasis of the phase symmetry and
each $e_{\mu}(x)$ has a unique weight. Note that phase symmetry preserves 
the Hodge decomposition.

One additional important fact is  that on the invariant ring $R_0^Q$ there exists a natural invariant pairing  turning it into a Frobenius algebra
~\cite{Dub, AKBA}:

\be
\eta_{\mu\nu} = \Res \frac{e_{\mu}(x) \, e_{\nu}(x)}{\prod_i \pd_i W_0(x)}.
\ee
Up to an irrelevant constant for the monomial basis it is $\eta_{\mu\nu} = \delta_{\mu+\nu, \rho}$.
This pairing
plays a crucial role in our construction.\\

Let us  introduce a couple of differentials~\cite{Saito} on differential forms on $\CC^5: 
\; D_{\pm} = \dd \pm \dd W_0(x) \wedge$.
They define the cohomology groups $H^*_{D_{\pm}}(\CC^5)$. The cohomologies are only nontrivial
in the top dimension $H^5_{D_{\pm}}(\CC^5) \overset{J}{\simeq} R_0$. The isomorphism $J$ has an explicit description
\be
J(e_{\mu}(x)) = e_{\mu}(x) \, \dd^5 x, \; e_{\mu}(x) \in R_0.
\ee 
We see, that $Q = \ZZ_5$ naturally acts on $H^5_{D_{\pm}}(\CC^5)$ and $J$ sends
 the $Q$-invariant part $R_0^Q$ to $Q$-invariant subspace $H^5_{D_{\pm}}(\CC^5)_{inv}$. 
Therefore, the latter space obtains the Hodge structure  as well.
Actually, this Hodge structure  naturally corresponds  to the Hodge structure on $H^3(X)$.

 The complex conjugation  acts 
on $H^3(X)$ so that $\overline{H^{p,q}(X)} = H^{q,p}(X)$, in particular
 $\overline{H^{2,1}(X)} = H^{1,2}(X)$. Through the isomorphism  between
$R_0^Q$ and $H^3(X)$ the complex conjugation acts also on the elements of the
ring $R_0^Q\; $ as  $* e_{\mu}(x) = p_{\mu} e_{\rho - \mu}(x),$ where
$p_{\mu}$ is a constant to be determined. In particular, differential form built from $e_{\mu}(x) + p_{\mu} e_{\rho - \mu}(x)
\in H^3(X, \RR)$ is real and $p_{\mu} p_{\rho - \mu}=1$.

\section{Oscillatory representation and computation of $\sigma_{\mu}(\phi)$} \label{sec:osc}
Relative homology groups $H_5(\CC^5, W_0 = L,\;\mathrm{Re}L \to \pm \infty)$ have a natural 
pairing with $Q$-invariant cohomology groups $H^5_{D_{\pm}}(\CC^5)_{inv}$:
\be
\langle e_{\mu}(x)\dd^5 x, \Gamma^{\pm} \rangle = \int_{\Gamma^{\pm}}
 e_{\mu}(x) e^{\mp W_0(x)} \dd^5 x, \;
  H_5(\CC^5, W_0 = L,\;\mathrm{Re}L \to \pm \infty).
\ee
Using this we define two invariant homology 
groups~\footnote{We are grateful to V. Vasiliev for explaining to  us the details 
about these homology groups and their connection with the middle homology of $X$.} 
$\mathscr{H}_5^{\pm,inv}$ as
 quotient of $H_5(\CC^5, W_0 = L,\;\mathrm{Re}L \to \pm \infty)$ with respect to the subgroups orthogonal
to $H^5_{D_{\pm}}(\CC^5)_{inv}$. 
Now we introduce  basises $\Gamma^{\pm}_{\mu}$ in the homology groups  
$\mathscr{H}_5^{\pm,inv}$ using the duality with the basises 
in  $H^5_{D_{\pm}}(\CC^5)_{inv}$:
\be
\int_{\Gamma^{\pm}_{\mu}}
 e_{\nu}(x) e^{\mp W_0(x)} \dd^5 x = \delta_{\mu\nu}
\ee
and the corresponding periods
\ba \label{sig1}
&\sigma_{\alpha\mu}^{\pm}(\phi) := \int_{\Gamma^{\pm}_{\mu}}
 e_{\alpha}(x) e^{\mp W(x, \phi)} \dd^5 x, \\
 &\sigma_{\mu}^{\pm}(\phi) := \sigma_{0\mu}^{\pm}(\phi)
\ea
which are understood as series expansions in $\phi$ around zero.

 Periods 
$\sigma_{\mu}^{\pm}(\phi)$ satisfy the same differential equation as periods $\omega_{\mu}(\phi)$
of the holomorphic volume form on $X$. Moreover, these sets of periods span same subspaces as
functions of $\phi$. It follows, that we can define cycles $Q^{\pm}_{\mu} \in \mathscr{H}_5^{\pm,inv}$ such that
\be \label{qcyc}
\int_{Q^{\pm}_{\mu}}
  e^{\mp W(x, \phi)} \dd^5 x =  \int_{q_{\mu}}\Omega = \int_{Q_{\mu}}\frac{\dd^5 x}{W(x, \phi)}
\ee
and periods $\omega_{\alpha\mu}^{\pm}(\phi)$ are given by the integrals over these cycles analogous to~\eqref{sig1}.

With these notations the idea of computation of periods~\cite{BB}
\be \label{sigma1}
\sigma^{\pm}_{\mu}(\phi) = \int_{\Gamma^{\pm}_{\mu}} 
 e^{\mp W(x, \phi)} \, \dd^5 x
\ee
can be stated as follows. \\

 To explicitly compute $\sigma^{\pm}_{\mu}(\phi)$, first  we expand the exponent in the
 integral~\eqref{sigma1} in $\phi$ representing $W(x,\phi) = W_0(x) + \sum_s \phi_s e_s(x)$
\be \label{sigma2}
\sigma^{\pm}_{\mu}(\phi) = \sum_m  \left(\prod_s\frac{(\pm\phi_s)^{m_s}}{ m_s!}\right) \int_{\Gamma^{\pm}_{\mu}} 
 \prod_r e_{r}(x)^{m_r} \, e^{\mp W_0(x)} \, \dd^5 x,
\ee
where $m := \{m_s \}_{s}, \; m_s \ge 0$ denotes a multi-index of powers of $\psi_s$ in the expansion
above.
We note, that $\sigma^-_{\mu}(\phi) = (-1)^{|\mu|}\sigma^+_{\mu}(\phi),$
so we focus on $\sigma_{\mu}(\phi) := \sigma^{+}_{\mu}(\phi).$

For each of the summands in~\eqref{sigma2} the form
 $\prod_s e_{s}(x)^{m_s} \, \dd^5 x$ belongs to $H^5_{D_{\pm}}(\CC^5)_{inv},$
because it is $Q-$invariant.
Therefore, we can expand it in the
 basis $\{e_{\mu}(x) \, \dd^5 x\}_{\mu=1}^{\dim R_0^Q} \subset
H^5_{D_{\pm}}(\CC^5)_{inv}.$ Namely we always can find such a polynomial $4-$form $U,$ that
\be \label{sigma3}
 \prod_s e_{s}(x)^{m_s} \, \dd^5 x =
 \sum_{\nu} C_{\nu}(m) \, e_{\nu}(x) \, \dd^5 x + D_{+} U  ,
\ee
where $C_{\nu}(m)$ are uniquely determined as coefficients of the expansion of the lhs
in the basis $e_{\mu}(x) \, \dd^5 x$.
Therefore for the integral in~\eqref{sigma2} we obtain
\be \label{sigma4}
\int_{\Gamma^{\pm}_{\mu}} 
 \prod_s e_{s}(x)^{m_s} \, e^{\mp W_0(x)} \, \dd^5 x = C_{\mu}(m).
\ee

Writing~\eqref{sigma2} explicitly we have
\be \label{qsigma1}
\sigma_{\mu}(\phi) = \sum_m \left(\prod_s\frac{\phi_s^{m_s}}{ m_s!}\right) \int_{\Gamma^{+}_{\mu}} 
 \prod_{s,i}  x_i^{m_s s_i } \, e^{- W_0(x)} \, \dd^5 x.
\ee
Let $m_s s_i = 5 n_i + \nu_i, \; \nu_i < 5$. Therefore we want to expand 
\be
\prod_i x_i^{5 n_i + \nu_i} \, \dd^5 x =
 \sum_{\nu} C_{\nu}(m) \, e_{\nu}(x) \, \dd^5 x + D_{+} U.
\ee
Note that 
\begin{multline} \label{rec}
D_+ \left(\frac{1}{5}x_1^{5n + k-4} \, f(x_2, \cdots, x_5) \, \dd x_2 \wedge \cdots \wedge \dd x_5
\right) = \\ =
\left[x_1^{5n+k} + \left(n+\frac{k-4}{5}\right) x_1^{5(n-1) + k} \right]  \, f(x_2, \cdots, 
x_5) \, \dd^5 x
\end{multline}
Therefore in $D_+$ cohomology we have
\be \label{rec1}
\prod_i x_i^{5 n_i + \nu_i} \, \dd^5 x = 
-\left(n_1+\frac{\nu_1-4}{5}\right) x_1^{5(n_1-1) + \nu_1}
\prod_{i=2}^5 x_i^{5n_i+\nu_i} \, \dd^5 x, \; \nu_i < 5.
\ee
By induction we obtain
\be \label{rec1}
\prod_i x_i^{5 n_i + \nu_i} \, \dd^5 x = 
(-1)^{\sum_i n_i}\prod_i \left(\frac{\nu_i+1}{5}\right)_{n_i}
\prod_i x_i^{\nu_i} \, \dd^5 x, \; \nu_i < 5.
\ee
where $(a)_{n}=\Gamma(a+n)/\Gamma(a)$.

Using~\eqref{rec} once again, we see that if any $\nu_i = 4$ then the differential form is trivial and
 the integral is zero. Hence, rhs of~\eqref{rec1} is proportional to $e_{\nu}(x)$ and
 gives the desired 
expression. Plugging~\eqref{rec1} into~\eqref{qsigma1} and integrating over $\Gamma^+_{\mu}$ gives the answer
\be
\sigma_{\mu}(\phi)=\sigma_{\mu}^+(\phi) =
 \sum_{n_i\ge 0} \prod_i \left(\frac{\mu_i+1}{5}\right)_{n_i} 
\sum_{m \in \Sigma_n } \prod_s\frac{\phi_s^{m_s}}{ m_s!} ,
\ee
where 
\be
\Sigma_n = \{m\;|\;\sum_s m_s s_i = 5n_i+\mu_i\}
\ee
Further we will also use the periods with slightly different normalization, which turn
 out to be convenient
\be \label{sigmaquint}
\hat{\sigma}_{\mu}(\phi) = \prod_i \Gamma\left(\frac{\mu_i+1}{5}\right)\sigma_{\mu}(\phi)=
 \sum_{n_i\ge 0} \prod_i \Gamma\left(n_i+\frac{\mu_i+1}{5}\right) 
\sum_{m \in \Sigma_n } \prod_s\frac{\phi_s^{m_s}}{ m_s!}. 
\ee
\section{Computation of the K\"ahler potential} \label{sec:comp}
 Pick any basis $Q^{\pm}_{\mu}$ of cycles with integer or real coefficients as in~\eqref{qcyc}.
Then for the K\"ahler potential we have the formula
\be \label{eqK1}
e^{-K} = \omega^+_{\mu}(\phi) C^{\mu\nu} \overline{\omega^-_{\nu}(\phi)}
\ee
in which the matrix $C^{\mu\nu}$ is related with the  Frobenius pairing $\eta$ as
\be \label{eqeta}
\eta_{\alpha\beta} = \omega^+_{\alpha\mu}(0)C^{\mu\nu} \omega^-_{\beta\nu}(0).
\ee
The last expression is due to~\cite{CV, Chiodo}.
Let also $T^{\pm}$ be a coordinate change matrix
 $Q^{\pm}_{\mu} = (T^{\pm})^{\nu}_{\mu} \Gamma^{\pm}_{\nu}$. Then  $M = (T^-)^{-1}\overline{T^-}$ is a real structure matrix, that is $M\bar{M} = 1$ and by construction $M$ doesn't depend on the choice
of basis $Q^{\pm}_{\mu}.$ $M$ is only defined by our choice of $\Gamma^{\pm}_{\mu}$.

In~\cite{AKBA} we deduced from~\eqref{eqK1} and~\eqref{eqeta}  the formula
\be \label{eqK}
e^{-K(\phi)} = \sigma^+_{\mu}(\phi) \eta^{\mu\lambda} M^{\nu}_{\lambda} \overline{\sigma^-_{\nu}(\phi)} = \sigma_{\mu} A^{\mu\nu}
\overline{\sigma_{\nu}},
\ee
where $\eta^{\mu\nu} = \eta_{\mu\nu} = \delta_{\mu, \rho-\nu}$.
In that papers our method to compute the real structure matrix $M$ used the knowledge of the periods in some
basis $q_{\mu}$ computed using the residue formula and monodromy considerations. 
However, this method gives only 4 out of 204 linearly independent periods for the quintic threefold $X$.

Therefore we propose here a different  method to find $M$. \\

\begin{lemma} Inverse intersection
 matrix $A^{\mu\nu}$ in~\eqref{eqK} is diagonal. 
\end{lemma}
\begin{proof}
 We may extend the action of the phase symmetry group to the action $\aA$ on the
 parameter space $\{ \phi_s \}$ such
that $W = W_0 + \sum_s \phi_s e_s(x)$ is invariant under this new action. Each $e_s(x)$ has a
 unique weight under this group action. 

Action $\aA$ can be compensated using the coordinate tranformation and therefore
 is trivial on the moduli space of the quintic (implying that point $W=W_0$ is an
 orbifold point of the moduli space).
In particular, $e^{-K} = \int_X \Omega\wedge\bar{\Omega}$ is $\aA$  invariant.
Consider
\be
e^{-K} = \sigma_{\mu} A^{\mu\nu} \overline{\sigma_{\nu}}
\ee
as a series in $\phi_s, \; \overline{\phi_t}$ 
Each monomial has a certain weight under $\aA$ . 
For the series to be invariant, each monomial must have weight 0. 
But weight of $\sigma_{\mu} \overline{\sigma_{\nu}}$ equals
to $\mu - \nu$ and due to non-degeneracy of weights of $\sigma_{\mu}$
 only the ones with $\mu = \nu$ have weight zero.
\end{proof}
Thus,~\eqref{eqK} becomes
\be
e^{-K} = \sum_{\mu} A^{\mu} |\sigma_{\mu}(\phi)|^2.
\ee
Moreover, the matrix $A$ should be real and, because $A = \eta \cdot M, \; M\bar{M}=1$ and $\eta_{\mu\nu} = \delta_{\mu+\nu,\rho},$ we have
 \be \label{mon}
A^{\mu} \, A^{\rho - \mu} = 1.
\ee
\paragraph{Monodromy considerations}
To fix the remaining 102 real numbers $A^{\mu}$ we use monodromy invariance of $e^{-K}$ around 
$\phi_0=\infty.$ Fix some $t = (t_1,t_2,t_3,t_4,t_5), \; |t|=5$ and let $\phi_s|_{s\ne t,0} = 0,$ also consider only the first order in $\phi_t$. Then the condition that period $\sigma_{\mu}(\phi)$ contains only non-zero summands
 of the form $\phi_0^{m_0} \, \phi_t$ implies that $\mu = t + const\cdot (1,1,1,1,1)$ mod 5.
 For each $t$
from the table below the only such possibilities are $\mu = t$ and $\mu = \rho - t' = (3,3,3,3,3) - t',$
 where $t'$ denotes a vector obtained from $t$ by permutation 
(written explicitly in the table below) of its coordinates.

Therefore, in this setting~\eqref{eqK} becomes
\be
e^{-K} = \sum_{k=0}^3 a_k |\hat{\sigma}_{(k,k,k,k,k)}|^2 + a_t|\hat{\sigma}_t|^2 + a_{\rho-t'} |\hat{\sigma}_{\rho-t'}|^2 + O(\phi_t^2),
\ee
where we used periods $\hat{\sigma}$ from~\eqref{sigmaquint}, $a_t = A^t /\prod_i \Gamma((t_i+1)/5)^2$ and $a_k, \; k=0,1,2,3$ are already known~\cite{COGP}.
 This expression should be monodromy invariant. We consider the effect of 
the transport of $\phi_0$ around $\infty$. From the formula~\eqref{sigmaquint} we have
\ba
&F_1 = \hat{\sigma}_{k}(\phi_t, \phi_0) = g_t \phi_k \, F(a,b;a+b \, |\, (\phi_0/5)^5) + O(\phi_t^6), \\
&F_2 = \hat{\sigma}_{\rho-t'}(\phi_t, \phi_0) = g_{\rho-t'} \phi_t \, \phi_0^{1-a-b} \, F(1-a,1-b;2-a-b\,|\, (\phi_0/5)^5) + O(\phi_t^6), 
\ea
where $g_t, \; g_{\rho-t'}$ are some constants. Explicitly for all different labels t\\
\begin{center}
\begin{tabular}{| l | l | l |} \label{tab1}
  t & $\rho-t'$ & (a, b)  \\ \hline
  (2,1,1,1,0) & (3,2,2,2,1) & (2/5,2/5)  \\
  (2,2,1,0,0) & (3,3,2,1,1) & (1/5,3/5)  \\
  (3,1,1,0,0) & (0,3,3,2,2) & (1/5,2/5)  \\
  (3,2,0,0,0) & (1,0,3,3,3) & (1/5,1/5)  \\
\end{tabular}
\end{center}
and 
\be
F(a,b;c|z):=\frac{\Gamma(a)\Gamma(b)}{\Gamma(c)} ~_2F_1(a,b;c; z).
\ee
When $\phi_0$ goes around infinity
\be
\begin{pmatrix}
F_1 \\
F_2
\end{pmatrix} = 
B \cdot\begin{pmatrix}
F_1 \\
F_2
\end{pmatrix},
\ee
where (e.g.~\cite{GR})
\be
B = \frac{1}{i s(a+b) }\begin{pmatrix}
c(a-b) - e^{i\pi (a+b)} & 2 s(a)s(b) \\
2 e^{2\pi i (a+b)} s(a)s(b) & e^{\pi i(a+b)} [e^{2\pi i a} + e^{2\pi i b}-2]/2
\end{pmatrix}.
\ee
Here $c(x) = \cos(\pi x), \; s(x) = \sin(\pi x)$. It is straightforward to show the following
\begin{prop}
\be
a_t |\hat{\sigma}_t|^2 + a_{\rho-t'} |\hat{\sigma}_{\rho-t'}|^2 = 
a_t \prod_i \Gamma\left(\frac{t_{i}+1}{5}\right)^2 |\sigma_t|^2 +
 a_{\rho-t'}\prod_i \Gamma\left(\frac{4-t_{i}}{5}\right)^2 |\sigma_{\rho-t'}|^2 
\ee
is $B$-invariant iff $a_t = -a_{\rho-t'}$. 
\end{prop}
Due to symmetry we have $a_{\rho-t'} = a_{\rho-t}$ in each case. From~\eqref{mon} 
it follows that the product of the coefficients at $|\sigma_{\mu}|^2$ and
 $|\sigma_{\rho-\mu}|^2$ in the expression for $e^{-K}$ should be 1:
\be
A^{\rho-t'} \cdot A^t=a_{\rho-t'} \cdot a_t \prod_i \Gamma\left(\frac{t_{i}+1}{5}\right)^2 
\Gamma\left(\frac{4-t_{i}}{5}\right)^2= 1.
\ee
Due to reflection formula $a_t = \pm\prod_i\sin(\pi (t_i+1)/5)$ up to a common factor of $\pi$.
The sign turns out to be minus for K\"ahler metric to be positive definite in the origin.
Therefore 
\be
A^{\mu} = (-1)^{\deg (\mu)/5} \prod{\gamma\left(\frac{\mu_i+1}{5}\right)}.
\ee

 Finally the K\"ahler potential becomes 
\be \label{eqKq}
e^{-K(\phi)} = \sum_{\mu=0}^{203}(-1)^{\deg (\mu)/5} \prod{\gamma\left(\frac{\mu_i+1}{5}\right)}
 |\sigma_{\mu}(\phi)|^2, 
\ee
where $ \gamma(x) = \frac{\Gamma(x)}{\Gamma(1-x)}$.

\section{Real structure on the cycles $\Gamma^{\pm}_{\mu}$} \label{sec:real}

Let cycles $\gamma_{\mu} \in H_3(X)$ be the images of cycles $\Gamma^+_{\mu}$ under the isomorphism
$\mathscr{H}_5^{+,inv} \simeq H_3(X)$. \\

Complex conjugation sends $(2,1)$-forms to $(1,2)$-forms. Similarly it
extends to a mapping on the dual homology cycles $\gamma_{\mu}$. In the real basis of cycles a
 version of the formula $\eqref{eqK}$ takes an especially simple form, because the real structure
matrix $M$ becomes an identity.
\begin{lemma} \label{theorem:lem1} Conjugation of homology classes has the
following form: $ *\gamma_{\mu} = p_{\mu}\gamma_{\rho - \mu},$
where  $\rho = (3,3,3,3,3)$ is a unique maximal degree element in the Milnor ring.\\
\end{lemma}
\begin{proof} We perform a proof for the cohomology classes represented by differential forms.
For one-dimensional $H^{3,0}(X)$ and $H^{0,3}(X)$ it is obvious. Let 
\be
\Omega_{2,1} := e_{t}(x)  \, \chi^l_{\bar{i}} \, \Omega_{ljk} \in H^{2,1}(X).
\ee
Any element from $H^{1,2}(X)$ is representable by a degree 10 polynomial  $P(x)$ as follows from~\eqref{chi}  as
\be
\overline{\Omega_{2,1}} = \Omega_{1,2}:=P(x)  \, \chi^l_{\bar{i}} \ 
 \chi^m_{\bar{j}} \, \Omega_{lmk} \in H^{1,2}(X).
\ee

The group of phase symmetries modulo common factor acts by isomorphisms on $X$. Therefore,
it also acts on the differential forms. Lhs and rhs of the previous equation should have the same
weigth under this action,
and weight of the lhs is equal $-t$ modulo $(1,1,1,1,1)$.
It follows that $P(x) = p_{t} \, e_{\rho - t}(x)$ with some constant $p_t$.


\end{proof}
Using this lemma and applying the complex conjugation of cycles to the formula~\eqref{eqK}
to obtain
\be
e^{-K} = \sum_{\mu} A^{\mu} |\sigma_{\mu}|^2  = 
\sum_{\mu} p_{\mu}^2 A^{\mu} \;|\sigma_{\rho - \mu}|^2,
\ee
it follows that $A^{\mu} = \pm 1/p_{\mu}.$  Now formula~\eqref{eqKq} implies
\be\label{coeff}
p_{\mu} = \prod_i{\gamma\left(\frac{4-\mu_i}{5}\right)}.
\ee

\section{Conclusions} \label{sec:concl}

The  method for computing the K\"ahler potential on the CY moduli space from~\cite{AKBA} modified in 
this paper does not require knowledge of periods in some real homology basis. Instead, we 
use some simple monodromy considerations to fix the real structure matrix. Another possible interesting
method would be to determine this matrix by direct computation of coefficients~\eqref{coeff} of the
complex conjugation in the basis $e_{\mu}(x)$. In this paper we use our modified method 
to compute Weil--Peterson metric on the whole 101-dimensional complex structure  moduli space 
of the quintic threefold around the orbifold point~\eqref{eqKq}. 
Together with the computation of the moduli space geometry of the K\"ahler structures through the mirror map~\cite{COGP} it describes the Special geometry of all
Ricci flat deformations of CY metric in the region. \\

Though we present our result for the quintic threefold, our method should be applicable
to a bigger class of models, which are connected with Landau--Ginzburg description, in particular
hypersurfaces in toric varieties. At least in the case of the hypersurfaces in weighted projective
spaces, we can, in principle, compute the basis $\{e_{\mu}(x)\}$ of $R_0^Q$ such, that the
pairing $\eta$ is antidiagonal, and the periods $\sigma_{\mu}(\phi)$. Indeed, it reduces to 
Jacobi ideal computations. Using the connection
of the pairing $\eta_{\mu\nu}$ with the natural pairing in the cohomology $H^3(X)$ it is possible
to prove~\eqref{intdiag} in this generality. Then the whole computation of the K\"ahler potential
is reduced to finding of the coefficients $A_{\mu}$. One way to do it is to restrict the expression
to the different one-dimensional subspaces of the moduli space and to require the monodromy
 invariance of the K\"ahler potential, as we did in the section~\ref{sec:comp} for the quintic
 threefold. In general, monodromy invariance translates to properties of generalized
 hypergeometric functions in one variable.
 
The main problem of our method in general is to choose the convenient
starting point $W_0(x)$ such, that Jacobi ideal computations are not be too complicated. 
 We plan to address possible generalizations in details in the future
publications.

\paragraph{Acknowledgements}

We are grateful to M. Bershtein, S. Galkin, A. Givental, M. Kontsevich, D. Orlov and V. Vasiliev
 for the usefull discussions and valuable remarks.

\printbibliography

\end{document}